\documentclass[10pt,twocolumn,letterpaper]{article}

\usepackage{cvpr}              
\usepackage{amsmath,amssymb,amsfonts}
\usepackage{amsthm}
\usepackage{algorithm}
\usepackage{algorithmic}
\usepackage{multirow}
\usepackage{tabularx}
\usepackage{booktabs}

\newtheorem{myDef}{Definition}

\newtheorem{Proposition}{Proposition}
\usepackage[dvipsnames]{xcolor}

\definecolor{cvprblue}{rgb}{0.21,0.49,0.74}
\usepackage[pagebackref,breaklinks,colorlinks,citecolor=cvprblue]{hyperref}

\def\paperID{16975} 
\def\confName{CVPR}
\def\confYear{2024}

\title{Differentiable Information Bottleneck for Deterministic Multi-view Clustering}

\author{Xiaoqiang Yan, Zhixiang Jin$^*$, Fengshou Han, Yangdong Ye$^*$\\
School of Computer and Artificial Intelligence, Zhengzhou University\\
{\tt\small \{iexqyan, ieydye\}@zzu.edu.cn, \{zxjin, iefshan\}@gs.zzu.edu.cn}
}

\begin{document}
\maketitle
\begin{abstract}
	In recent several years, the information bottleneck (IB) principle provides an information-theoretic framework for deep multi-view clustering (MVC) by compressing multi-view observations while preserving the relevant information of multiple views. Although existing IB-based deep MVC methods have achieved huge success, they rely on variational approximation and distribution assumption to estimate the lower bound of mutual information, which is a notoriously hard and impractical problem in high-dimensional multi-view spaces. In this work, we propose a new differentiable information bottleneck (DIB) method, which provides a deterministic and analytical MVC solution by fitting the mutual information without the necessity of variational approximation. Specifically, we first propose to directly fit the mutual information of high-dimensional spaces by leveraging normalized kernel Gram matrix, which does not require any auxiliary neural estimator to estimate the lower bound of mutual information. Then, based on the new mutual information measurement, a deterministic multi-view neural network with analytical gradients is explicitly trained to parameterize IB principle, which derives a deterministic compression of input variables from different views. Finally, a triplet consistency discovery mechanism is devised, which is capable of mining the feature consistency, cluster consistency and joint consistency based on the deterministic and compact representations. Extensive experimental results show the superiority of our DIB method on 6 benchmarks compared with 13 state-of-the-art baselines.
\end{abstract}
\renewcommand{\thefootnote}{}
\footnotetext{$^*$ Corresponding author.}

\section{Introduction}

Multi-view clustering (MVC)~\cite{10108535} aims to discover hidden patterns or potential structures by leveraging the complementary information in multi-view data. In the literature, MVC involving traditional machine learning techniques can be classified into graph-based~\cite{9860070}, subspace-based~\cite{9358980} and matrix factorization-based methods~\cite{Liu_2021_ICCV}. However, these traditional MVC based on shallow learning models often exhibits poor representation ability on large-scale high-dimensional and non-linear multi-view data.
Recently, deep learning models have seen widespread adoption in MVC owing to their powerful representation capability, resulting in  deep MVC~\cite{DBLP:conf/ijcai/LiWTGY19,DBLP:journals/corr/abs-2205-03803,DBLP:journals/tnn/WangTXGCJ23,DBLP:journals/tnse/LinDWG23,DBLP:conf/cvpr/XuT0P0022,chen2023deep,DBLP:conf/cvpr/YanZLT0LL23,DBLP:journals/tkde/XuRTYPYPYH23}. Although achieving promising performance, most existing deep MVC emphasizes the relevant correlations across multiple views and ignores the limitations of the irrelevant information in each view, such as noises, corruptions or even view-private attributes.

In addressing these challenges, several recent approaches have resorted to the notable information bottleneck (IB) principle to multi-view clustering~\cite{DBLP:conf/ijcai/LouYY13,xu2014large}. By formulating mutual information (MI), IB provides an information-theoretic framework to learn a compact representation and  remove irrelevant information for a given task~\cite{tishby1999information}. Despite the successful applications, the estimation of mutual information is a notoriously hard problem in high-dimensional multi-view space since the complicated joint distribution of two variables is often criticized to be hard or impossible. To overcome this constraint, variational approximation offers a natural solution to construct and estimate a lower bound of the mutual information of high-dimensional variables~\cite{DBLP:conf/iclr/AlemiFD017,wang2023self}. 
Inspired by this, IB and its variational versions have achieved promising performance by exploring the training dynamics in deep multi-view clustering and representation models as well as a learning objective~\cite{DBLP:conf/aaai/MaoYGY21,10313078,yan2023cross}. 
However, the variational approximation in existing IB-based deep MVC results in the uncertainty in multi-view representation learning.
Specifically, existing IB-based deep MVC leverages variational approximation to estimate the marginal and posterior probability distribution of the potential feature representations (as shown in Fig.1). In the process of variational approximation, IB-based deep MVC methods introduce an auxiliary neural network to estimate the mean and variance of the posterior distribution so as to fit the posterior distribution. Then, they impose $Kullback$-$Leibler$ (KL) divergence constraint between the posterior distribution and variational approximation to align them~\cite{DBLP:conf/iclr/AlemiFD017}. Thus, the approximation error introduced by variational approximation increases the uncertainty of mutual information estimation.

\begin{figure}[t]
\begin{center}
	\includegraphics[width=0.9\linewidth]{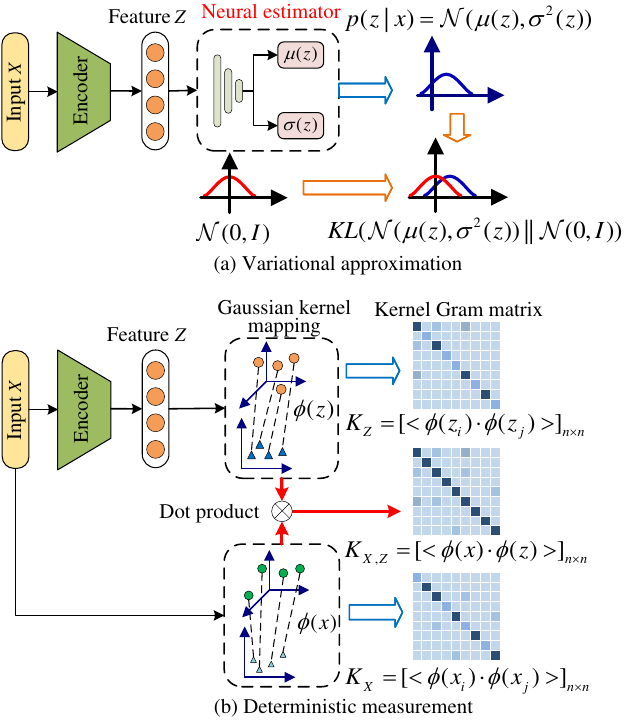}
\end{center}
\caption{Variational approximation and deterministic measurement. (a) Variational approximation requires a neural estimator to estimate the posterior distribution $p(z|x)$ of the representation while assuming the marginal distribution of the representation follows a standard normal distribution, so as to approximate the lower bound of the mutual information. (b) Our deterministic measurement leverages the gaussian kernel function to construct the kernel Gram matrix which can measure the distance between data pairs. Then the eigenvalues of the Gram matrix can be expressed to entropy function (see section~\ref{MIWVA} for detailed proof). }\label{deterministic}
\end{figure}

In this study, we propose a novel differentiable information bottleneck (DIB) method for deterministic and analytical multi-view clustering without variational approximation. As shown in Figure~\ref{model}, DIB learns a latent and compact space for each view in a deterministic compression manner while capturing triplet consistency derived from high-level features and semantic labels across multiple views. To this end, we first design a novel MI measurement to directly fit the mutual information between high-dimensional multi-view spaces by leveraging normalized kernel Gram matrix, which can measure the information about feature representations directly from the original data and does not need any neural estimators to learn the lower bound of mutual information. Then, based on the proposed MI measurement without variational approximation, a deterministic multi-view neural network is explicitly trained to parameterize IB principle with analytical gradients, which derives a deterministic compression and learns a compact representation for each view. Finally, a unified objective function under our DIB framework is devised to optimize the deterministic compression and triplet consistency discovery simultaneously, in which consistent information of multiple views from high-level features and semantic labels is characterized based on the deterministic and compact representations. 
Extensive experimental results verify the effectiveness and promising performance of DIB compared with state-of-the-art baselines. In summary, this study makes the following contributions.
\begin{itemize}
\item We propose a novel differentiable information bottleneck (DIB) method for deterministic multi-view clustering, which provides a deterministic and analytical MVC solution by essentially fitting the mutual information without the necessity of variational approximation.
\item A novel MI measurement without variational approximation is proposed to fit the mutual information of high-dimensional spaces directly by eigenvalues of the normalized kernel Gram matrix. This work is a valuable attempt to directly measure the information about feature representations from the data rather than building a neural estimator to approximate the lower bound of mutual information.
\item A deterministic neural network with analytical gradients is built to parameterize IB principle, which enjoys a concise and tractable objective and provides a deterministic compression of input variables from different views.
\end{itemize}

\section{Related Work and Preliminaries}

\subsection{Information Bottleneck}
The information bottleneck (IB)~\cite{tishby1999information} originates from rate-distortion and attempts to compress source variable $\textbf{X}$ into its compressed representation $\textbf{Z}$ while preserving the information that can predict relevant variable $\textbf{Y}$. It is assumed that we have access to the joint distribution $p(\textbf{X},\textbf{Y})$ with the goal of pursuing the following quantization
\begin{equation}\label{IB-func}
	\max{\textbf{IB}}_{\beta}=I(\textbf{Z};\textbf{Y}) - \beta I(\textbf{X};\textbf{Z})
\end{equation}
where $I()$ is the mutual information and $\beta$ is a trade-off parameter.

Recently, the idea of exploring a good representation with IB principle is becoming prevalent and it also achieves great success in  deep multi-view learning, such as multi-view clustering~\cite{DBLP:conf/aaai/MaoYGY21,10313078,yan2023cross}, multi-view representation learning~\cite{DBLP:conf/iclr/Federici0FKA20,DBLP:conf/aaai/WanZZH21,DBLP:journals/corr/abs-2204-12496}, multi-view graph clustering~\cite{DBLP:journals/corr/abs-2210-07011}.  However, both the ``black box" operation of neural network and the approximate error introduced by variational approximation increase the uncertainty of mutual information estimation.

Different from the aforementioned approaches, DIB provides a deterministic and analytical MVC solution by essentially fitting the mutual information without the necessity of variational approximation. Moreover, the new measurement of mutual information is differentiable and can explicitly parameterize IB principle with analytical gradients.

\subsection{Deep Multi-view Clustering}
Existing deep MVC approaches can be classified into the following categories: deep embedding-based, deep graph-based, deep adversarial-based and contrastive MVC. Deep embedding-based MVC jointly optimizes the embedded representation of multiple views and the clustering process~\cite{DBLP:journals/corr/abs-2205-03803,DBLP:journals/isci/XuRLPZX21}. Deep graph-based MVC learns the cluster structures of multi-view data by forming a better graph from multiple views~\cite{DBLP:journals/tnse/LinDWG23,DBLP:journals/isci/ZhaoYN23}. Deep adversarial-based MVC uses adversarial training as a regularizer to align the multi-view data~\cite{DBLP:conf/ijcai/LiWTGY19}. Contrastive MVC enables a better latent space of multiple views by characterizing the positive and negative samples~\cite{DBLP:conf/aaai/Li0LPZ021,DBLP:conf/cvpr/XuT0P0022}.

The proposed DIB is remarkably different from existing deep MVC approaches. First, DIB constructs a view-specific encoder with the constraint of differentiable mutual information, which can learn a compact and discriminative representation for each view by preserving relevant information and eliminating irrelevant information simultaneously. Second, DIB leverages a deterministic neural network with analytical gradients driven by the mutual information without variational approximation to parameterize IB principle, which enjoys a concise and tractable objective and provides a deterministic compression of input variables from different views. Third, a triplet consistency discovery mechanism under our DIB framework is devised, which capture the feature consistency, cluster consistency and joint consistency in a triplet manner.

\section{Differentiable Information Bottleneck}
\subsection{Problem Statement}
\emph{Problem statement}. Given an unlabelled multi-view collection $\{\textbf{X}^v \in \mathbb{R}^{N\times D^v}\}_{v=1}^V$, multi-view clustering aims to partition the data samples into $K$ clusters, where $V$ is the number of views, $x_i^v\in \mathbb{R}^{D^v}$ is the samples of the $v$-th view, $N$ and $D^v$ are the data size and feature dimension of the $v$-th view respectively.

Recently, the deep MVC approaches based on IB principle have achieved huge success since it provides an information-theoretic framework to learn a compact representation and remove irrelevant information for a given task. However, despite the successful applications, the approximate error introduced by variational approximation increases the uncertainty of mutual information estimation. Aiming at these issues, we propose a novel differentiable information bottleneck for deterministic MVC. Intuitively, DIB should meet the following requirements. 1) \textbf{Information measurement}. It should directly measure the information of original data about its feature representations and does not need any neural estimators to learn the lower bound of mutual information. 2) \textbf{Deterministic compression}. A neural network driven by the information measurement without variational approximation should have analytical gradients that allow us to parameterize the IB principle and optimize it through backward propagation. 3) \textbf{Consistency maximization}. DIB should characterize the consistency of multiple views more comprehensively based on the deterministic and compact representations.
To facilitate these goals, we design the network architecture of DIB method as shown in Fig. 2. From this figure, we can see that DIB consists of deterministic compression and triplet consistency discovery. For convenience, we first provide the definition of the proposed DIB method.

\begin{myDef}[Differentiable information bottleneck, DIB]
	Suppose there exists an unlabelled multi-view collection $\{\textbf{X}^v \in \mathbb{R}^{N\times D^v}\}_{v=1}^V$, DIB consists of deterministic compression and triplet consistency discovery. The deterministic compression part learns a deterministic and compact representation $\{\textbf{Z}^v\}_{v=1}^V$ for each view $\{\textbf{X}^v\}_{v=1}^V$ using view-specific encoder $E^v$ with an information-theoretic constraint. In the triplet consistency discovery part, we explore the consistency of multiple views from high-level features $\{\textbf{H}^v\}_{v=1}^V$  and semantic labels $\{\textbf{S}^v\}_{v=1}^V$ in a triplet manner. In summary, the goal of DIB is to search a reasonable clustering assignment $\textbf{C}$ by learning a deterministic and compact representation for each view while maximally preserving the consistency across multiple views.
\end{myDef}

\subsection{Mutual Information without Variational Approximation}\label{MIWVA}
In this subsection, we design a novel mutual information measurement without variational approximation to directly fit the mutual information of high-dimensional spaces by leveraging normalized kernel Gram matrix. Specifically, we first show the eigenvalues of the kernel Gram matrix can be expressed by the recently proposed R$\acute{e}$nyi's $\alpha$-order entropy~\cite{DBLP:books/sp/P2010}. Then, the mutual information can be achieved via matrix-based R$\acute{e}$nyi's $\alpha$-order entropy function and the joint-entropy function without variational approximation. 
For convenience, we first provide the definition of Gram matrix.
\begin{myDef}[Gram matrix]
	Given a set of vectors $\{v_i\}_{i=1}^n$ in an inner product space, the Gram matrix $G$ is defined as an $n\times n$ matrix with entries
	\begin{equation}\label{entropy}
		\scriptsize
		\begin{aligned}
			& G_{ij}=<v_i, v_j> \\
		\end{aligned}
	\end{equation}
	where $<v_i, v_j>$ denotes the inner product of vectors $v_i$ and $v_j$.
\end{myDef}

Different from approximating the lower bound of mutual information through a neural estimator in variational approximation~\cite{DBLP:conf/iclr/AlemiFD017}, we utilize the eigenvalues of the kernel Gram matrix to directly fit the recently proposed R$\acute{e}$nyi's $\alpha$-order entropy as shown in following proposition.

\begin{figure}[t]
	\begin{center}
		\includegraphics[width=0.9\linewidth]{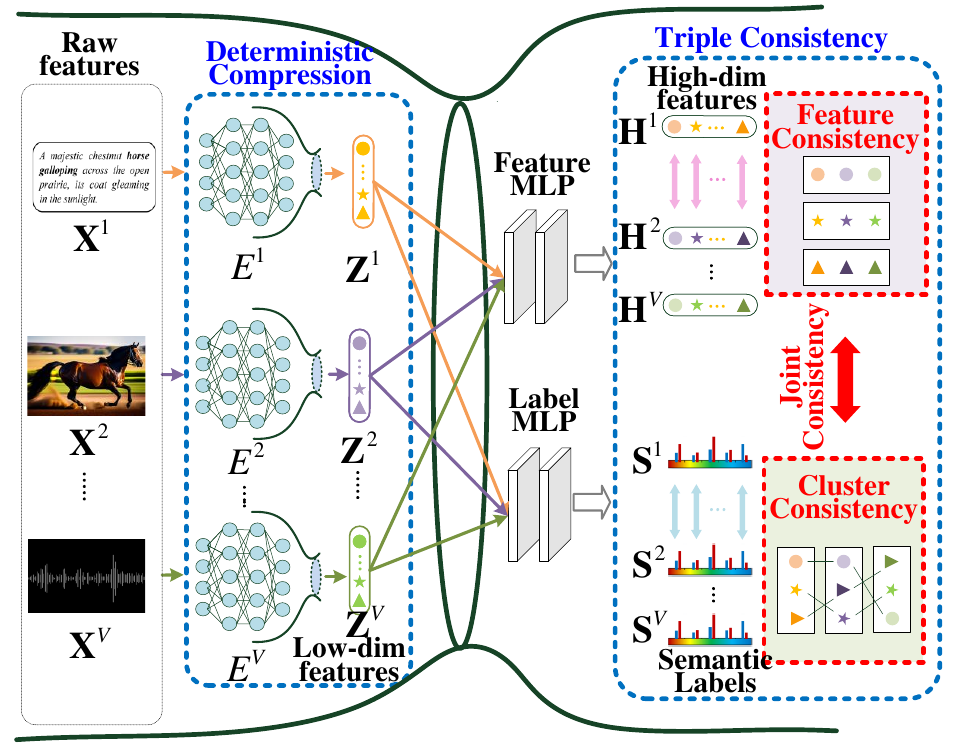}
	\end{center}
	\caption{Framework of the DIB. In DIB, the deterministic compression aims to learn a compact representation for each view through the new mutual information measurement without variational approximation. The triplet consistency discovery mechanism is devised to mine the feature, cluster and joint consistency from the compact representation.}\label{model}
\end{figure}

\begin{Proposition}\label{directlyFit}
	The R$\acute{e}$nyi's $\alpha$-order entropy of a random variable $\textbf{X} \in \mathcal{R}^{N\times D^v}$ can be fitted by the eigenvalues of a Gram matrix which is constructed by evaluating a positive definite kernel function for each pair of data points.
\end{Proposition}
\begin{proof}
	According to ~\cite{renyi1961measures}, the R$\acute{e}$nyi's $\alpha$-order entropy of a random variable $\textbf{X}$ can be defined as follows
	\begin{equation}\label{basicRenyi}
		\footnotesize
		\begin{aligned}
			& H_\alpha (\textbf{X}) = \frac{1}{1-\alpha} \log_2\int_{\mathcal{X}}p^\alpha(x)dx\\
		\end{aligned}
	\end{equation}
	where $\alpha \in (0,1) \cup (1,\infty)$, $p(x)$ is the probability density function of the random variable $\textbf{X}$.
	
	To better understanding, we take $\alpha = 2$ as an example to illustrate how the the eigenvalues of a Gram matrix fit the R$\acute{e}$nyi's 2-order entropy. For $\alpha = 2$, we can leverage the Parzen density estimator~\cite{jenssen2006cauchy} with a Gaussian kernel $g_\sigma(x, y)=\text{exp}\big(-\frac{1}{2\sigma^2}||x-y||^2\big)$ to calculate the probability density function $p(x)$, i.e., $\hat{p}(x)=\frac{1}{n}\sum_{i=1}^ng_\sigma(x, x_i)$, which can be plugged into the Eq.~\ref{basicRenyi}, yields
	\begin{equation}\label{densityEstimator}
		\footnotesize
		\begin{aligned}
			\hat{H}_2(\textbf{X})&=-\text{log}_2\int_\mathcal{X}\hat{p}^2(x)dx \\
			&=-\text{log}_2\int_\mathcal{X}\biggl(\frac{1}{n}\sum_{i=1}^ng_\sigma(x,x_i)\biggl)^2dx \\
			&=-\text{log}_2\biggl(\frac{1}{n^2}\sum_{j=1}^n\sum_{i=1}^n\int_\mathcal{X}g_\sigma(x,x_j)g_\sigma(x,x_i)dx\biggl) \\
			&=-\text{log}_2\frac{1}{n^2}\sum_{i=1}^n\sum_{j=1}^ng_{\sigma}(x_i,x_j)
		\end{aligned}
	\end{equation}
	
	According to Eq.~\ref{densityEstimator}, the R$\acute{e}$nyi's 2-order entropy~\cite{DBLP:journals/tit/GiraldoRP15} of a random variable $\textbf{X}$ can be characterized in terms of the eigenvalues of a Gram matrix $G$, where $G_{ij}=g_\sigma(x_i,x_j)$. Then, the Eq.~\ref{densityEstimator} can be rewritten as
	\begin{equation}
		\footnotesize
		\begin{aligned}
			\hat{H}_2(\textbf{X})&= -\text{log}_2\big(\frac{1}{n^2}\text{tr}(G^TG)\big)
		\end{aligned}
	\end{equation}
	
	The derivation of the R$\acute{e}$nyi's $\alpha$-order entropy from its 2-order version can be found in Supplementary A, so $H_\alpha(\textbf{X})$ can be calculated directly by the eigenvalues of a Gram matrix $G$ which obtained by evaluating a positive definite kernel function $g_\sigma$ for each pair of data points, but without considering intermediate steps in density estimation.
\end{proof}

To facilitate the calculation of mutual information, we provide the definition of matrix-based R$\acute{e}$nyi's $\alpha$-order entropy and joint entropy function as follows.

\begin{myDef}[Matrix-based R$\acute{e}$nyi's $\alpha$-order entropy function]
	Given a data collection $\textbf{X}^v \in \mathcal{R}^{N\times D^v}$ from the $v$-th view, where $N$ is the number of data samples. The Gram matrix $G_\mathbf{x}$ is obtained by calculating a positive definite kernel function $g_\mathbf{x}$ for all data pairs, i.e., $G_\mathbf{x}(i,j)=g_\mathbf{x}(x^v_i, x^v_j)$. The matrix-based R$\acute{e}$nyi's $\alpha$-order entropy of $X^v$ can be defined as follows
	\begin{equation}\label{entropy}
		\scriptsize
		\begin{aligned}
			& H_{\alpha} (A^v_\mathbf{x}) = \frac{1}{1-\alpha} \log_{2}(tr((A^v_\mathbf{x})^\alpha)) = \frac{1}{1-\alpha} \log_{2} \big(\sum_{i=1}^N \lambda_i(A^v_\mathbf{x})^{\alpha}\big) \\
		\end{aligned}
	\end{equation}
	where $\alpha \in (0,1) \cup (1,\infty)$, $A^v_\mathbf{x} = \frac{G_\mathbf{x}}{tr(G_\mathbf{x})}$ is normalized from the Gram matrix $G_\mathbf{x}$, $\lambda_i(A^v_\mathbf{x})$ indicates the $i$-th eigenvalue of $A^v_\mathbf{x}$.
\end{myDef}

\begin{myDef} [Matrix-based R$\acute{e}$nyi's $\alpha$-order joint-entropy function]
	Given a data collection $\textbf{X}^v \in \mathcal{R}^{N\times D^v}$ from the $v$-th view and its corresponding representation $\textbf{Z}^v \in \mathcal{R}^{N\times D^v}$, the matrix-based R$\acute{e}$nyi's $\alpha$-order joint-entropy can be defined as follows
	\begin{equation}\label{joint-entropy}
		\footnotesize
		\begin{aligned}
			& H_{\alpha} (A^v_\mathbf{x},A^v_\mathbf{z}) = H_{\alpha}\left( \frac{A^v_\mathbf{x} \circ A^v_\mathbf{z}}{tr(A^v_\mathbf{x}\circ A^v_\mathbf{z})} \right)\\
		\end{aligned}
	\end{equation}
	where $A^v_\mathbf{x} \circ A^v_\mathbf{z}$ indicates the Hadamard product between $A^v_\mathbf{x}$ and $A^v_\mathbf{z}$.
\end{myDef}

Given Eq.~\ref{entropy} and Eq.~\ref{joint-entropy}, the MI between high-dimensional variables can be directly calculated as follows
\begin{equation}\label{DMI}
	\footnotesize
	\begin{aligned}
		I_{\alpha} (\textbf{X}^v;\textbf{Z}^v) = H_{\alpha}(A^v_\mathbf{x}) + H_{\alpha}(A^v_\mathbf{z}) - H_{\alpha}(A^v_\mathbf{x},A^v_\mathbf{z})\\
	\end{aligned}
\end{equation}

In next subsection, we prove that the novel MI measurement without variational approximation has analytical gradients that allow us to parameterize the IB principle and optimize it as an objective.

\subsection{Deterministic Compression}
IB aims to compress data observations and preserve the relevant information for a given task. Recently, it has been applied to analyze and understand the learning dynamics of DNNs~\cite{DBLP:conf/itw/TishbyZ15} benefiting by the progress of MI neural estimators, such as variational approximation~\cite{DBLP:conf/icml/BelghaziBROBHC18}. However, existing MI neural estimators need the explicit estimation of the underlying distributions of data (more details can be found in Supplementary A),
which often results in the uncertainty in representation learning. In this study, we design a novel MI measurement without variational approximation (Eq.~\ref{DMI}), which is capable of directly fitting the MI of high-dimensional spaces.
However, its differentiate property is unclear, which impedes its practical deployment as a loss function to parameterize IB principle. Next, we prove the MI measurement without variational approximation in Eq.~\ref{DMI} has an analytical gradient.

\begin{Proposition}\label{differentation}
	Given a data collection $\textbf{X}^v \in \mathcal{R}^{N\times D^v}$ from the $v$-th view and its corresponding representation $\textbf{Z}^v \in \mathcal{R}^{N\times D^v}$, the mutual information measurement in Eq.~\ref{DMI} has an analytical gradient.
\end{Proposition}
\begin{proof}
	First, we present the gradient of $H_{\alpha}(A^v_\mathbf{x})$, which can be calculated as follows
	\begin{equation}
		\footnotesize
		\begin{aligned}
			& \frac{\partial H_{\alpha}(A^v_\mathbf{x})}{\partial A^v_\mathbf{x}} = \frac{\alpha}{1-\alpha} \frac{(A^v_\mathbf{x})^{\alpha-1}}{(1-\alpha)tr((A^v_\mathbf{x})^{\alpha})}\\
		\end{aligned}
	\end{equation}
	
	Similarly, the gradient of $H_{\alpha} (A^v_\mathbf{x},A^v_\mathbf{z})$ can be calculated as follows
	\begin{equation}
		\footnotesize
		\begin{aligned}
			& \frac{\partial H_{\alpha}(A^v_\mathbf{x},A^v_\mathbf{z})}{\partial A^v_\mathbf{x}} = \frac{\alpha}{1-\alpha} \left[ \frac{(A^v_\mathbf{x}\circ A^v_\mathbf{z})^{\alpha-1} \circ A^v_\mathbf{z}}{tr(A^v_\mathbf{x}\circ A^v_\mathbf{z})^\alpha} - \frac{I\circ A^v_\mathbf{z}}{tr(A^v_\mathbf{x}\circ A^v_\mathbf{z})} \right]\\
		\end{aligned}
	\end{equation}
	
	Since $I_{\alpha} (\textbf{X}^v;\textbf{Z}^v) = H_{\alpha}(A^v_\mathbf{x}) + H_{\alpha}(A^v_\mathbf{z}) - H_{\alpha}(A^v_\mathbf{x},A^v_\mathbf{z})$, $I_{\alpha} (\textbf{X}^v;\textbf{Z}^v)$ has an analytical gradient.
\end{proof}

In practice, the gradient of $I_{\alpha} (\textbf{X}^v;\textbf{Z}^v)$ can be computed using automatic differentiation libraries like Tensorflow and PyTorch.

In summary, based on \textbf{Proposition}~\ref{directlyFit}, DIB can fit the mutual information from the original data and feature representation directly. Based on \textbf{Proposition}~\ref{differentation}, the MI measurement without variational approximation has analytical gradients that allow us to parameterize the IB principle and optimize it as an objective. Thus, we construct a view-specific encoder $E^v$ with analytical gradients for each view to parameterize the IB principle by directly fitting the MI $I_\alpha(\textbf{X}^v; \textbf{Z}^v)$ between the original view data $\{\textbf{X}^v\}_{v=1}^V$ and representations $\{\textbf{Z}^v\}_{v=1}^V$, i.e., $\textbf{Z}^v=E^v(\textbf{X}^v)$, which derives a deterministic compression of input variables from from different views. Note that, the MI measurement in Eq.~\ref{DMI} do not need any neural estimators to explicit estimate the underlying distribution of data, which enables us to parameterize IB principle with a deterministic neural network. Thus, we can obtain the loss function of deterministic compression inspired by IB principle as follows
\begin{equation}\label{compress}
	\footnotesize
	\begin{aligned}
		& \text{min }\mathcal{L}_{comp}=\sum_{v=1}^V I_\alpha(\textbf{X}^v;\textbf{Z}^v) \\
	\end{aligned}
\end{equation}

\subsection{Triplet Consistency Discovery}
In MVC scenarios, another key issue is to capture the consistency across multiple views. Generally, multiple views of a data sample are different in attributes or input forms but show consistency in high-level features and semantics, which is also the foundation
for effective MVC. Based on this observation, we first transform the compact representation of each view into high-level feature and cluster spaces by MLP. Then, a triplet consistency discovery mechanisms designed to mine the consistency across views from high-level features, clusters and joint of features and clusters.

Firstly, we design a \textbf{feature consistency} to make high-level features $\{\textbf{H}^v\}_{v=1}^V$ focus on learning the common features across multiple views, which can be characterized through the popular contrastive learning~\cite{10313078,DBLP:conf/cvpr/He0WXG20}. Specifically, each high-level feature $h_i^v$ can form $(VN-1)$ feature pairs $\{h_i^v, h_j^m\}_{j=1,\dots,N}^{m=1,\dots,V}$ with all features except itself, where $\{h_i^v, h_i^m\}_{m \neq v}$ denotes $(V-1)$ positive feature pairs and $\{h_i^v, h_j^m\}_{j=1,\dots,N}^{m=1,\dots,V} - \{h_i^v, h_i^m\}_{m \neq v}$ denotes $V(N-1)$ negative feature pairs. In contrastive learning, the goal is to maximize the similarities between positive pairs while minimizing those of negative pairs. Then, the similarity between two features can be measured by the cosine distance as follows
\begin{equation}
	\footnotesize
	\begin{aligned}
		d(h_i^v, h_j^m)=\frac{<h_i^v, h_j^m>}{||h_i^v||\ ||h_j^m||}
	\end{aligned}
\end{equation}
where $<\cdot,\cdot>$ represents the dot product operator. And then, the feature consistency objective $\mathcal{L}_H$ between high-level features $\{\textbf{H}^v\}_{v=1}^V$ can be formulated as
\begin{equation}\label{contrastive}
	\footnotesize
	\begin{aligned}
		\text{max }\mathcal{L}_{fea} &= \sum_{v=1}^V\sum_{m \neq v}I(\textbf{H}^v;\textbf{H}^m) \\
		&\approx \sum_{v=1}^V\sum_{m \neq v}\mathbb{E}\left[\text{log}\frac{e^{d(h_i^v, h_i^m)}}{\sum_{d(h_i^v,h_j^k)\in Neg}^Ne^{d(h_i^v,h_j^k)}}\right]\\
		&+V(V-1)\text{log }N
	\end{aligned}
\end{equation}
where $Neg$ denotes negative feature pairs, and $d(h_i^v,h_j^k)$ is the similarity of negative feature pairs.


\begin{algorithm}[t!]
	\small
	\caption{Differentiable Information Bottleneck}
	\label{DIB}
	\begin{algorithmic}[1]
		\STATE{\bfseries Input:} Multi-view data $\{\mathbf{X}^v\}_{v=1}^V$, cluster number $K$, iteration number $I_t$.
		\STATE \textbf{Random initialization:} Initialize encoder $E^v$;
		\FOR{epoch $\in \{0,1,2,\cdots,I_t\}$}
		\STATE Obtain the representations $\{\mathbf{Z}^v\}_{v=1}^V$ via $\mathbf{Z}^v=E^v(\mathbf{X}^v)$.
		\STATE Obtain the high-level features $\{\mathbf{H}^v\}_{v=1}^V$ and semantic labels $\{\mathbf{S}^v\}_{v=1}^V$ through the feature MLP and label MLP, respectively.
		\STATE Calculate the feature, cluster consistency loss by Eq.~\ref{contrastive} and calculate the joint consistency by Eq.~\ref{hsconsistency}.
		\STATE Calculate the compress loss function by Eq.~\ref{compress}.
		\STATE Update the parameters of the whole model by back propagation.
		\ENDFOR
		\STATE {\bfseries Output:} The clustering results $\textbf{C}$.
		\smallskip
	\end{algorithmic}
\end{algorithm}

Secondly, we can achieve the \textbf{cluster consistency} by contrastive learning to ensure the identical cluster labels convey consistent high-level semantics across views. Similarly, the cluster consistency objective $\mathcal{L}_{clu}$ between semantic labels $\{\textbf{S}^v\}_{v=1}^V$ also can be calculated by Eq.~\ref{contrastive} (the detailed formulation for $\mathcal{L}_{clu}$ can be found in Supplementary A). 

Finally, we design a \textbf{joint consistency} between high-level features $\{\textbf{H}^v\}_{v=1}^V$ and cluster assignments $\{\textbf{S}^v\}_{v=1}^V$ to further refine the consistency across views. Intuitively, for one data instance, the learned feature representations from different views should maximally preserve the consistency for its cluster labels. This is to say, the mutual information between high-level features and cluster assignments also should be maximized for better clustering. Thus, we define the joint consistency between high-level features $\{\textbf{H}^v\}_{v=1}^V$ and cluster assignments $\{\textbf{S}^v\}_{v=1}^V$ as follows
\begin{equation}\label{hsconsistency}
	\footnotesize
	\begin{aligned}
		\text{max }\mathcal{L}_{joint} = \sum_{v=1}^VI(\textbf{H}^v;\textbf{S}^v)
	\end{aligned}
\end{equation}

DIB consists of deterministic compression and triplet consistency discovery. Similar to IB principle, we construct a trade-off between deterministic compression and triplet consistency discovery as follows

\begin{equation}\label{overall}
	\footnotesize
	\begin{aligned}
		& \mathcal{L}_{overall} = \underbrace{ \max (\mathcal{L}_{fea} + \mathcal{L}_{clu} + \gamma \mathcal{L}_{joint} ) }_{Consistency} + \beta \underbrace{ \min \mathcal{L}_{comp} }_{Compression} \\
	\end{aligned}
\end{equation}
where $\gamma$ and $\beta$ are the trade-off parameters that control the impact of joint consistency and deterministic compression on the final clustering performance. $\mathcal{L}_{fea}$ and $\mathcal{L}_{clu}$ can be calculated by Eq.~\ref{contrastive} as in contrastive clustering~\cite{10313078}, $\mathcal{L}_{joint}$ and $\mathcal{L}_{comp}$ can be calculated by the proposed MI measurement without variational approximation as in Eq.~\ref{DMI}. The DIB is presented in Algorithm~\ref{DIB}.

\section{Experiments}
\subsection{Datasets}
The proposed DIB is evaluated on six widely-used multi-view datasets. 
MNIST-USPS~\cite{DBLP:conf/icml/0001HLZZ19} contains 5,000 samples of handwritten digit images based on two of features distributed across 10 categories.
Berkeley drosophila genome project (BDGP)~\cite{DBLP:journals/bioinformatics/CaiWHD12} contains 2,500 samples of Drosophila embryos belonging to 5 categories, each represented by visual and textual views. Handwritten~\cite{hull1994database} is a popular handwritten character dataset containing 2000 samples drawn by 6 different handwriting styles in 10 categories. Event segmentation and prediction (ESP)~\cite{DBLP:conf/aaaiss/AhnD05} is designed for action recognition, which captures 11,032 samples from 4 different viewpoints with different sensors or cameras, such as RGB cameras and depth sensors. Flickr~\cite{cox2008flickr} is a widely used set of images that contains 12,154 samples from three shooting perspectives from different users at different locations and times, organized into seven categories. For Caltech~\cite{fei2004learning}, we construct three datasets with different numbers of views from on Caltech to evaluate the proposed method. Specifically, Caltech-3V contains WM, CENTRIST and LBP; Caltech-4V includes WM, CENTRIST, LBP and GIST; and Caltech-5V encompasses WM, CENTRIST, LBP, GIST and HOG.
\begin{table*}[t]
	\caption{Clustering performance of all methods on the six datasets. Bold and underline indicate the best and second best results.}\label{comparison}
	\footnotesize
	\begin{center}
		\setlength{\tabcolsep}{2.0mm}
		\begin{tabular}{r|ccc|ccc|ccc|ccc}
			\hline \multicolumn{1}{c|}{Datasets} & \multicolumn{3}{c|}{ MNIST-USPS } & \multicolumn{3}{c|}{ BDGP } & \multicolumn{3}{c|}{ Handwritten } & \multicolumn{3}{c}{ ESP } \\
			\hline \multicolumn{1}{c|}{Evaluation metrics} & ACC & NMI & PUR & ACC & NMI & PUR & ACC & NMI & PUR & ACC & NMI & PUR \\
			\hline BMVC (TPAMI'2019) & 72.22 & 60.37 & 73.50 & 85.68 & 71.98 & 85.68 & 84.45 & 77.59 & 84.45 & 47.95 & 32.01 & 50.33 \\
			SMKKM (ICCV'2021) & 73.35 & 64.58 & 74.20 & 63.72 & 51.22 & 64.66 & $\underline{88.80}$ & 80.88 & $\underline{88.80}$ & 48.76 & 31.23 & 49.37 \\
			OPLFMVC (ICML'2021) & 68.38 & 60.57 & 68.38 & 66.44 & 42.96 & 66.44 & 77.65 & 74.86 & 80.15 & 51.80 & 32.08 & 51.80 \\
			FastMICE (TKDE'2023) & 90.73 & 90.24 & 91.44 & 77.42 & 62.81 & 77.91 & 84.55 & 86.68 & 85.78 & 54.94 & 36.32 & 55.29 \\\hline
			MFLVC (CVPR'2022) & 99.58 & 98.72 & 99.58 & 98.60 & 95.87 & 98.60 & 82.48 & 82.15 & 82.48 & 56.02 & 36.52 & 56.02 \\
			CVCL (ICCV'2023) & 99.58 & 98.81 & 99.58 & $\underline{98.88}$ & $\underline{96.28}$ & $\underline{98.88}$ & 78.10 & 80.77 & 81.70 & 47.05 & 31.80 & 48.64 \\
			AONGR (INS'2023) & 99.36 & 98.23 & 99.36 & 92.04 & 82.47 & 92.04 & 80.30 & 80.34 & 80.40 & 50.68 & 36.77 & 51.27 \\
			GCFAgg (CVPR'2023) & 96.28 & 93.04 & 96.28 & 96.52 & 91.74 & 96.52 & 51.75 & 54.02 & 55.70 & 57.61 & \textbf{40.59} & 57.61 \\
			SDMVC (TKDE'2023) & $\underline{99.82}$ & $\underline{99.47}$ & $\underline{99.82}$ & 96.80 & 92.00 & 96.80 & 77.63 & $\underline{86.92}$ & 77.63 & 49.57 & 36.16 & 49.57 \\ \hline
			DMIM (AAAI'2021) & 98.12 & 97.53 & 98.09 & 93.18 & 92.63 & 93.49 & 81.23 & 83.74 & 82.83 & 56.11 & 37.26 & 55.42 \\
			DMIB (TCYB'2022) & 96.71 & 97.12 & 97.35 & 96.57 & 95.20 & 96.32 & 80.92 & 81.66 & 81.15 & 51.03 & 23.17 & 50.68 \\
			ConGMC (TMM'2023) & 99.01 & 98.45 & 98.62 & 97.28 & 94.36 & 95.51 & 83.65 & 84.29 & 84.57 & $\underline{58.45}$ & 37.62 & $\underline{58.37}$ \\
			DCIB (TNNLS'2023) & 56.67 & 72.38 & 56.84 & 61.01 & 45.46 & 61.80 & 68.60 & 79.48 & 68.60 & 54.40 & 36.18 & 54.40 \\ \hline
			DIB (ours) & \textbf{99.86} & \textbf{99.56} & \textbf{99.86} & \textbf{99.00} & \textbf{96.65} & \textbf{99.00} & \textbf{88.95} & \textbf{89.92} & \textbf{88.95} & \textbf{59.06} & $\underline{37.77}$ & \textbf{59.06} \\
			\hline
		\end{tabular}
	\end{center}
	
	\footnotesize
	\begin{center}
		\setlength{\tabcolsep}{2.0mm}
		\begin{tabular}{r|ccc|ccc|ccc|ccc}
			\hline \multicolumn{1}{c|}{ Datasets } & \multicolumn{3}{c|}{ Flicker } & \multicolumn{3}{c|}{ Caltech-3V } & \multicolumn{3}{c|}{ Caltech-4V } & \multicolumn{3}{c}{ Caltech-5V }\\
			\hline \multicolumn{1}{c|}{ Evaluation metrics } & ACC & NMI & PUR & ACC & NMI & PUR & ACC & NMI & PUR & ACC & NMI & PUR \\
			\hline BMVC (TPAMI'2019) & 56.73 & 36.04 & 56.80 & 64.93 & 53.46 & 45.38 & 73.36 & 69.74 & 73.36 & 77.29 & 70.96 & 77.29 \\
			SMKKM (ICCV'2021) & 59.31 & 38.77 & \underline{59.42} & 51.45 & 38.92 & 29.75 & 68.83 & 55.58 & 69.96 & 73.54 & 64.60 & 73.54 \\
			OPLFMVC (ICML'2021) & 51.32 & 31.08 & 51.49 & 57.21 & 42.56 & 36.93 & 74.29 & 53.73 & 74.29 & 79.14 & 65.00 & 79.14 \\
			FastMICE (TKDE'2023) & 54.75 & 35.50 & 54.99 & 64.23 & 57.25 & 50.19 & 73.90 & 66.88 & $\underline{77.00}$ & 78.63 & 72.48 & 78.63 \\ \hline
			MFLVC (CVPR'2022) & 53.98 & 36.97 & 54.60 & 60.77 & 56.52 & 61.71 & 61.75 & 64.18 & 62.00 & 71.37 & 66.57 & 71.37 \\
			CVCL (ICCV'2023) & 57.75 & 38.43 & 57.75 & 66.14 & 58.29 & 66.29 & 72.32 & 63.04 & 74.64 & 81.01 & 70.42 & 81.01 \\
			AONGR (INS'2023) & 54.34 & 38.52 & 54.64 & 53.86 & 52.09 & 57.57 & 59.93 & 59.29 & 64.50 & 65.71 & 61.14 & 67.93 \\
			GCFAgg (CVPR'2023) & 31.18 & 19.58 & 37.85 & 59.43 & 55.72 & 59.71 & 48.86 & 48.40 & 53.29 & 50.93 & 55.05 & 54.64 \\
			SDMVC (TKDE'2023) & 39.30 & 19.04 & 39.31 & 67.66 & 57.72 & 50.52 & $\underline{74.79}$ & 68.03 & $\textbf{77.79}$ & $\underline{83.84}$ & $\underline{78.08}$ & $\underline{83.84}$ \\ \hline
			DMIM (AAAI'2021) & 57.44 & 34.83 & 57.68 & 70.71 & 58.67 & 69.34 & 73.07 & $\underline{70.94}$ & 74.09 & 79.28 & 63.09 & 80.16 \\
			DMIB (TCYB'2022) & 55.74 & 27.88 & 56.29 & 71.21 & 59.23 & 70.52 & 72.78 & 66.82 & 71.95 & 82.28 & 68.02 & 81.97 \\
			ConGMC (TMM'2023) & $\underline{60.05}$ & 37.92 & 57.06 & $\underline{73.37}$ & $\underline{64.80}$ & $\underline{74.54}$ & 73.78 & 69.14 & 75.25 & 83.78 & 73.55 & 82.70 \\
			DCIB (TNNLS'2023) & 58.78 & $\underline{38.88}$ & 58.78 & 58.40 & 51.50 & 58.40 & 69.70 & 60.90 & 69.72 & 75.15 & 68.74 & 75.63 \\ \hline
			DIB (ours) & $\textbf{60.40}$ & $\textbf{40.23}$ & $\textbf{60.40}$ & $\textbf{74.71}$ & $\textbf{68.46}$ & $\textbf{75.71}$ & \textbf{75.64} & \textbf{71.51} & 76.64 & $\textbf{84.79}$ & $\textbf{78.34}$ & $\textbf{84.79}$ \\
			\hline
		\end{tabular}
	\end{center}
\end{table*}

\subsection{Implementation}
The encoders in the DIB are composed of a four-layer fully connected network. The feature MLP consists of two linear layers. The label MLP is constructed by a linear layer and a Softmax layer. The DIB is implemented through Pytorch's public toolbox. We use the Adam optimizer for optimization and set the learning rate to $3 \times 10^{-4}$. The parameters $\alpha$ and $\beta$ in the loss function (Eq.~\ref{overall}) are fixed, i.e., ${\alpha = 0.01}$ and ${\beta = 0.01}$, for all used datasets. We implement the experiments on a Windows 11 platform and an NVIDIA 4060Ti GPU with 16G of RAM.

\subsection{Baselines}
We compare the DIB with the three types of state-of-the-art methods. 1) Traditional MVC: binary MVC (BMVC)~\cite{DBLP:journals/pami/ZhangLSSS19}, simple multi-kernel $k$-means (SMKKM)~\cite{DBLP:conf/iccv/LiuZ0TW0Z21}, one-pass late fusion MVC (OPLFMVC)~\cite{DBLP:conf/icml/Liu0LWZTT0Z21} and fast MVC via ensembles (FastMICE)~\cite{DBLP:journals/tkde/HuangWL23}).
2) Deep MVC: multi-feature multi-level clustering (MFLVC)~\cite{DBLP:conf/cvpr/XuT0P0022},  cross-view contrastive
learning (CVCL)~\cite{chen2023deep}, auto-weighted orthogonal and nonnegative graph reconstruction (AONGR)~\cite{DBLP:journals/isci/ZhaoYN23}, global and cross-view feature
aggregation (GCFAgg)~\cite{DBLP:conf/cvpr/YanZLT0LL23} and self-discriminative MVC (SDMVC)~\cite{DBLP:journals/tkde/XuRTYPYPYH23}). 3) IB-based deep MVC: deep mutual information maximin (DMIM)~\cite{DBLP:conf/aaai/MaoYGY21}, deep multi-view information bottleneck (DMIB) \cite{DBLP:conf/iclr/Federici0FKA20}, consistency-guided multi-modal clustering (ConGMC)~\cite{10313078}  and deep correlated information bottleneck (DCIB)~\cite{yan2023cross}).
The parameter settings of the baselines in our experiments are fine-tuned for each dataset according to the descriptions in the corresponding papers. 

\subsection{Evaluation Metrics}
We use three widely-used clustering metrics including clustering accuracy (ACC)~\cite{DBLP:journals/tcyb/LiuLXGWW13}, normalised mutual information (NMI)~\cite{DBLP:journals/jmlr/StrehlG02} and purity (PUR)~\cite{DBLP:journals/isci/Aliguliyev09} to quantify the clustering results. The reported results of the used algorithms are the average values by running 10 times.

\subsection{Performance Analysis}
We evaluate the effectiveness of the DIB with traditional, deep and IB-based MVC baselines. The comparison results are shown in Table ~\ref{comparison}. From this table, we obtain the following observations: 1) The DIB outperforms the traditional MVC, which demonstrates its superior ability of representation learning of high-dimensional space compared with traditional shallow MVC baselines. 2) Compared with several latest SOTA deep MVC baselines, DIB also achieves better performance. For example, the DIB obtains 3.04\%, 12.01\%, 8.38\%, 1.46\% and 9.49\% improvements compared with MFLVC, CVCL, AONGR, GCFAgg and SDMVC on the ESP dataset in terms of ACC metric. This is mainly because that the deterministic compression in DIB can learn a discriminative and compact representation for each view. 3) Compared with IB-based deep MVC baselines, DIB achieves the best results on all evaluation metrics in the used datasets. This is mainly because that the DIB can directly measure the information about
feature representations from the source data rather than building a neural estimator to approximate the lower bound of mutual information. Besides, we conduct a significance test \cite{DBLP:journals/pami/SunCDLDY22} to verify that the performance of the DIB is statistically better than the representative baselines (see Supplementary B for more details).


\begin{table*}[t]
	\caption{Ablation experiments on loss components.}\label{ablation}
	\begin{center}
		\footnotesize
		\begin{tabular}{c|cccc|ccc|ccc|ccc}
			\hline \multicolumn{1}{c|}{ } & \multicolumn{4}{c|}{ Loss Components } & \multicolumn{3}{c|}{ MNIST-USPS } & \multicolumn{3}{c|}{ BDGP } & \multicolumn{3}{c}{ ESP } \\
			\hline & $\mathcal{L}_{clu}$ & $\mathcal{L}_{fea}$ & $\mathcal{L}_{comp}$ & $\mathcal{L}_{joint}$ & ACC & NMI & PUR & ACC & NMI & PUR & ACC & NMI & PUR \\
			\textbf{A)} & $\checkmark$ & & &                           & 85.72  & 90.24 & 85.72 & 82.58 & 87.13 & 82.58 & 48.01 & 31.96 & 49.43 \\
			\textbf{B)} & $\checkmark$ & $\checkmark$ & &              & 92.96  & 93.48 & 52.96 & 91.00 & 90.23 & 92.00 & 49.95 & 35.38 & 53.18 \\
			\textbf{C)} & $\checkmark$ & $\checkmark$ & & $\checkmark$ & 96.98  & 97.67 & 96.98 & 95.24 & 91.33 & 92.24 & 49.83 & 35.12 & 52.80 \\
			\textbf{D)} & $\checkmark$ & $\checkmark$ & $\checkmark$ & & 99.84  & 99.50 & 99.84 & 98.72 & 95.69 & 98.72 & 53.30 & 37.40 & 54.86 \\
			\textbf{E)} & $\checkmark$ & $\checkmark$ & $\checkmark$ & $\checkmark$ & \textbf{99.86} & \textbf{99.56} & \textbf{99.86} & \textbf{99.00} & \textbf{96.65} & \textbf{99.00} & \textbf{59.06} & \textbf{37.77} & \textbf{59.06} \\
			\hline
		\end{tabular}
	\end{center}
\end{table*}

\subsection{Ablation Study}

To verify the effectiveness of each components, we consider the following five scenarios: A) Retain $\mathcal{L}_{clu}$. In this case, we only use the cluster consistency. B) Retain $\mathcal{L}_{clu}$ and $\mathcal{L}_{fea}$. We use the cluster consistency and feature consistency simultaneously. C) Retain $\mathcal{L}_{clu}$, $\mathcal{L}_{fea}$, and $\mathcal{L}_{joint}$. In this scenario, we add joint consistency between cluster consistency and feature consistency so that the consistency learned from high-level features can further optimize the cluster structure. D) Retain $\mathcal{L}_{clu}$, $\mathcal{L}_{fea}$ and $\mathcal{L}_{comp}$. We add the deterministic compression component to scenario B). E) Retain $\mathcal{L}_{clu}$, $\mathcal{L}_{fea}$, $\mathcal{L}_{comp}$ and $\mathcal{L}_{joint}$. In this case, we perform MVC with the overall loss function of the DIB.

\begin{figure}[t]
	\center
	\small
	\includegraphics[width=0.4\linewidth]{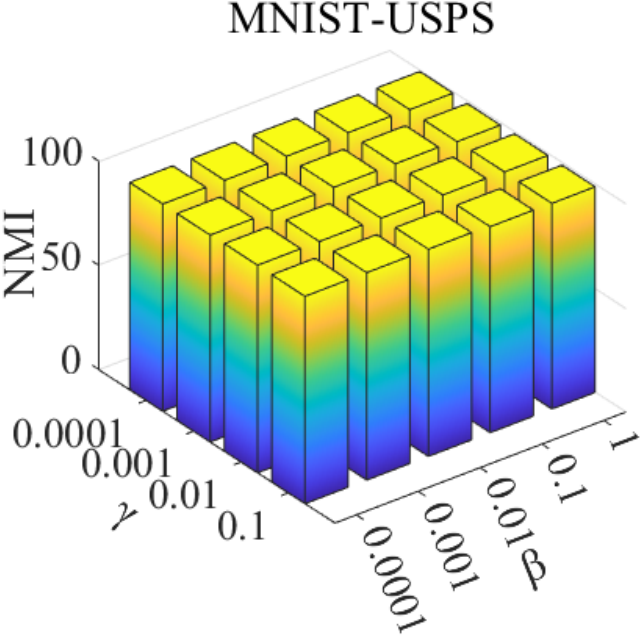}
	\includegraphics[width=0.4\linewidth]{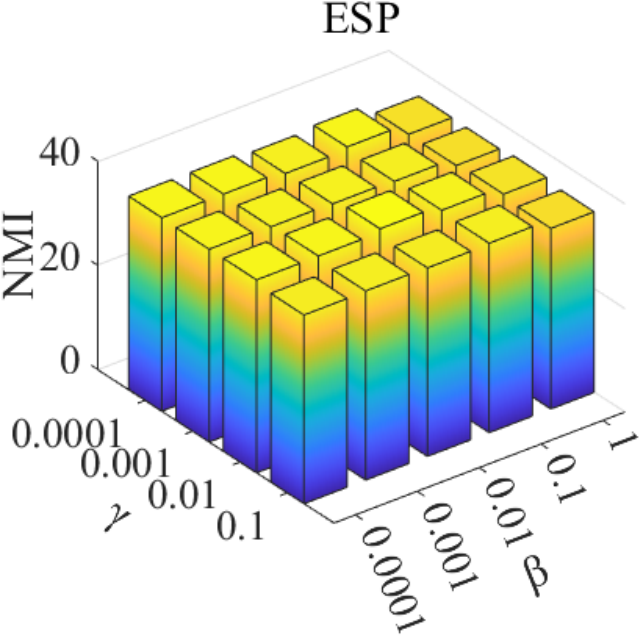}
	\caption{Parameter $\gamma$ and $\beta$ sensitivity experiment results.}\label{sanwei}
\end{figure}
From Table ~\ref{ablation}, we can obtain the following observations. According to A) and B), we can find that learning the consistency from high-level features can improve the clustering performance. This shows that it makes sense to map low-level features to high-level features to learn the common semantics. 
According to B) and D), we can observe that the deterministic compression can significantly improve the clustering performance. According to D) and E), we can find that establishing triplet consistency with the deterministic compression component can further enhance the clustering performance. This indicates that the learned deterministic and compact representation can facilitate the consistency discovery across multiple views. According to C) and D) and E), we can get that removing deterministic compression or discarding triplet consistency will make the final clustering performance degraded, which suggests that the two main parts of our model are highly integrated and refined.

Besides, we replace the differentiable MI measurement in DIB with the variational approximation. The correspondingly experimental results further verify the effectiveness of the proposed MI measurement without variational approximation (see Supplementary B for details).

\begin{figure}[t]
	\center
	\small
	\includegraphics[width=0.4\linewidth]{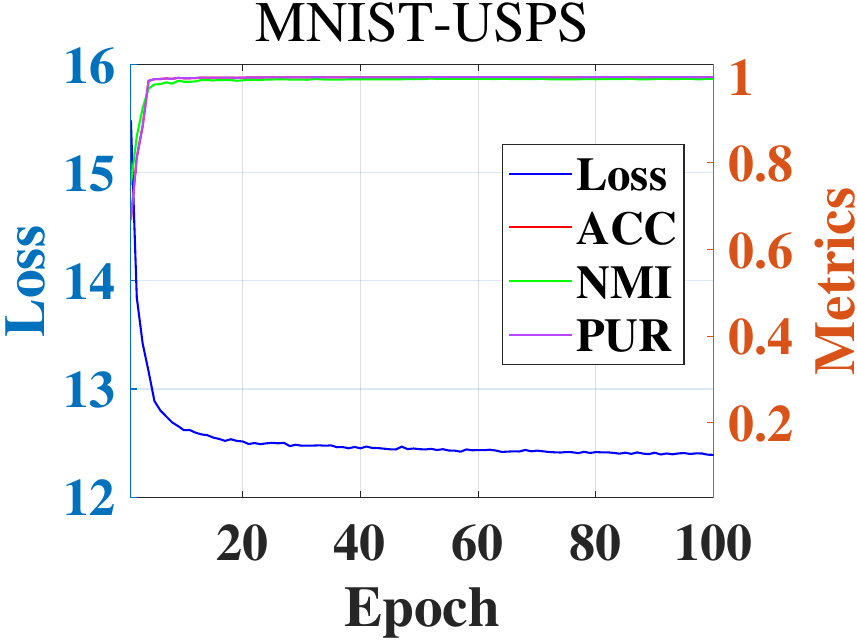}
	\includegraphics[width=0.4\linewidth]{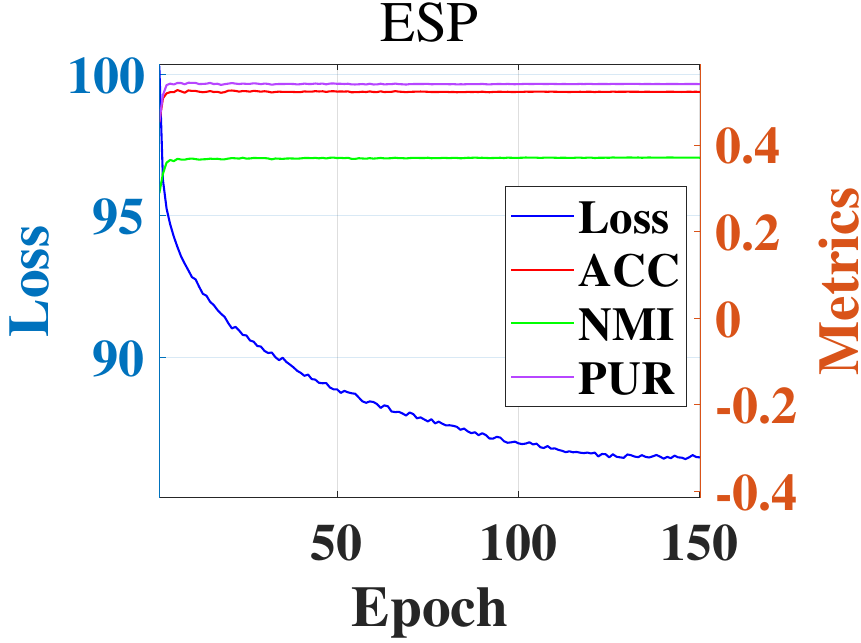}
	\caption{Convergence curves on MNIST-USPS and ESP.}\label{zhexian}
\end{figure}

\subsection{Parameter Sensitivity Analysis}
In this subsection, we evaluate the impact of the trade-off parameters ${\gamma}$ and ${\beta}$ on the clustering performance of the DIB on two representative datasets (MNIST-USPS and ESP). Specifically, we investigate the values of $\gamma$ and $\beta$ in the range of  ${\{10^{-4},10^{-3}, 10^{-2},10^{-1}\}}$ and ${\{10^{-4},10^{-3},10^{-2},10^{-1},10^0\}}$, respectively. From Figure~\ref{sanwei}, we can observe that the DIB method can achieve stable clustering performance with different combinations of parameters in MNIST-USPS and ESP datasets. This suggests that our model is insensitive to the choice of ${\gamma}$ and ${\beta}$. Based on the experimental results, we set ${\gamma = 0.01}$ and ${\beta = 0.01}$ for all used datasets in this study.


\subsection{Convergence Analysis}

\begin{figure}[t]
	\center
	\small
	\includegraphics[width=0.77\linewidth]{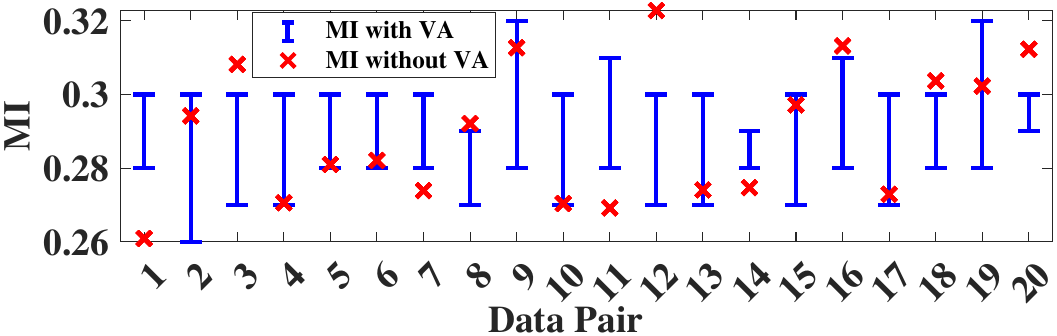}
	\caption{ Mutual information with/without VA on the data pairs sampled from MNIST-USPS.}\label{uncertainty}
\end{figure}

Figure~\ref{zhexian} reports the values of the loss function and the evaluation metrics of the DIB algorithm as the iterations increase. It can be observed that the loss function and evaluation metrics of the DIB approach to a fixed point with the epochs increase. This phenomenon shows that our DIB algorithm enjoys a good convergence property.
%

\subsection{MI Measurement Evaluation}
To verify the effectiveness of our MI without variational approximation (VA), we compare the MI with VA with our MI without VA by sampling 20 data pairs from MNIST-USPS dataset randomly. As shown in Figure~\ref{uncertainty}, we observe that MI with VA fluctuates within a range and our MI without VA is a definite value.

\section{Conclusions and Future Work}
This paper investigates a novel differentiable information bottleneck method, which provides a deterministic and analytical MVC solution by essentially fitting the mutual information without the necessity of variational approximation. It is a valuable attempt to directly measure the information about feature representations from the data. In future, it is interesting to use the mutual information without variation approximation to conduct layer-by-layer training for a DNN.
{
    \small
    \bibliographystyle{ieeenat_fullname}
    \bibliography{egbib}
}


\end{document}